\documentclass[12pt]{extarticle}
\pdfoutput=1

\usepackage[ngerman,english]{babel}
\usepackage{graphicx}
\usepackage{framed}
\usepackage[normalem]{ulem}
\usepackage{xfrac}
\usepackage{enumerate}
\usepackage{comment}
\usepackage[utf8]{inputenc}
\usepackage[top=1 in,bottom=1in, left=1 in, right=1 in]{geometry}

\usepackage{amsmath}
\usepackage{amsfonts}
\usepackage{amsthm}
\usepackage{amssymb}
\usepackage{mathtools}
\usepackage{hyperref}
\usepackage[utf8]{inputenc} 
\usepackage[T1]{fontenc}    
\usepackage{enumitem}
\usepackage{cleveref}

\usepackage[colorinlistoftodos]{todonotes}

\theoremstyle{definition}
\newtheorem{definition}{Definition}

\theoremstyle{plain}
\newtheorem{theorem}[definition]{Theorem}

\newtheorem{lemma}[definition]{Lemma}
\newtheorem{proposition}[definition]{Proposition}

\theoremstyle{remark}
\newtheorem{remark}[definition]{Remark}

\usepackage[autostyle]{csquotes}
\usepackage{biblatex}
\usepackage{authblk}

\title{Complexity Theory meets Ordinary Differential Equations}

\author[1]{Adalbert Fono}
\author[1]{Noah Wedlich}
\author[2,3,4,5]{Holger Boche}
\author[1,6,7,8]{Gitta Kutyniok}
\affil[1]{Department of Mathematics, Ludwig-Maximilians-Universität München, Germany}
\affil[2]{Institute of Theoretical Information Technology, TUM School of Computation, Information and Technology, Technical University of Munich, Germany}
\affil[3]{Munich Quantum Valley (MQV)}
\affil[4]{Munich Center for Quantum Science and Technology (MCQST)}
\affil[5]{Cluster of Excellence “Centre for Tactile Internet with Human-in-the-Loop” (CeTI) of Technische Universität Dresden, Germany }

\affil[6]{Department of Physics and Technology, University of Tromsø, Norway}
\affil[7]{Munich Center for Machine Learning (MCML), Munich, Germany}
\affil[8]{DRL-German Aerospace Center, Germany}
\date{} 

\begin{document}

\maketitle

\begin{abstract}
This contribution investigates the computational complexity of simulating linear ordinary differential equations (ODEs) on digital computers. We provide an exact characterization of the complexity blowup for a class of ODEs of arbitrary order based on their algebraic properties, extending the characterization of first order ODEs shown in \cite{Boche2021ComplexityBlowup}.
Complexity blowup indeed arises in most ODEs (except for certain degenerate cases) and means that there exists a low complexity input signal, which can be generated on a Turing machine in polynomial time, leading to a corresponding high complexity output signal of the system in the sense that the computation time for determining an approximation up to $n$ significant digits grows faster than any polynomial in $n$. Similarly, we derive an analogous blowup criterion for a subclass of first-order systems of linear ODEs. 
Finally, we discuss the implications for the simulation of analog systems governed by ODEs and exemplarily apply our framework to a simple model of neuronal dynamics---the leaky integrate-and-fire neuron---heavily employed in neuroscience.
\end{abstract}

\section{Introduction}

The study of differential equations is one of the most important areas of mathematics, with applications ranging from physics and engineering to biology and Artificial Intelligence. It allows us to model and analyze dynamical systems, describing how quantities change over time or space. The mathematical theory of differential equations, concerned with finding (explicit) solutions, proving their existence and uniqueness under certain conditions, or deducing their qualitative properties, is a well-studied (though not completed) field. 

However, with the advent of modern computational methods, the numerical solution, emulation, and simulation of differential equations have become increasingly important. For instance, the design and development of analog systems, as well as optimization testing, strongly depend on numerical simulations \cite{BauerLecler91, JainYangYoshino91, UesakaKazuyoshiKawamata20, BrunoMesquitaAntonio08, SanzSanchoRoldan07, BrambillaAngeloAmore01, AlbuquerqueSilva05, ChenWeiweiHan14}. Similarly, emulation of such systems as in neuroscience, e.g., the behaviour of neural networks as well as single neurons and synapses, is a key step in advancing knowledge and exploiting the properties of (biological) brains \cite{gerstner_kistler_naud_paninski_2014, Christensen2022NCSurvey, Mehonic2024Neuromorphic}. 
Understanding the feasibility and computational complexity of these tasks is therefore practically relevant. 

One obstacle in this analysis is the fact that analog systems are inherently dependent on continuous (physical) quantities but are almost always simulated on digital computers. Digital processing relies on manipulating a discrete set of values in a step-by-step manner, i.e., digital computers employ rational numbers to approximate the exact continuous quantities. Although discretizing the inputs and parameters of an analog system to obtain a rational description is typically not a severe restriction, the true output of the rational analog system is nevertheless often irrational. For instance, even for (linear) ordinary differential equations with rational coefficients, the corresponding solution may take irrational values when evaluated at rational time points. Hence, the output of the (digital) simulation and the analog system operate in different domains. 

This dichotomy is an often-neglected aspect in the simulation of analog systems on digital hardware and the subsequent mathematical analysis assessing the correctness and accuracy of the simulation. However, there are certain restrictions under which analog systems can actually be simulated on digital machines with error control. Here, not only the aspired simulation is carried out, but additionally, the correctness of the simulation up to an arbitrary but fixed error is guaranteed. The described property illustrates the problem of the \emph{computability} of a task. In fact, this condition cannot be satisfied, among others, for the solution of the wave equation \cite{PourElNingZhong06}, the Fourier transform \cite{PourElRichards83, BocheMoenich20}, and spectral factorization \cite{BochePohl19}. This means that for a specific problem, such as the wave equation with rational coefficients, there exist computable initial conditions such that the corresponding solution does not take computable values for certain rational inputs. 

Going beyond the computability of a task, a similarly important question is the one of \emph{complexity}, i.e., assuming that a system is computable in the above sense, what is the computational complexity for achieving a simulation with a given maximum approximation error? Complexity is understood here as the necessary number of computational steps to achieve the desired accuracy (which, by the computability assumption, is feasible in principle). Of particular interest is the occurrence of a \emph{complexity blowup} describing the situation where a low-complexity input to the simulation of an analog system entails a high-complexity computation to obtain the associated output. 

More formally, assume that we want to simulate a given system $S: u\mapsto y$ mapping an input signal $u$ onto an output signal $y$ on a digital computer. For a specific input $u_s$ and associated output $y_s = S(u_s)$, the number of computational steps to generate an approximation up to an error of $2^{-n}$ of $u_s$ and $y_s$ represents the complexity/computation time of the computation of the input and output, respectively. Classically, in computer science, low complexity refers to tasks where the computation time grows at most polynomially in the relevant parameter, with the understanding that polynomial-time problems are feasible to solve in practice, whereas super-polynomial problems are typically infeasible. Hence, computational complexity is crucial when assessing the practical feasibility of the implementation of a specific problem. Here, we will consider systems $S$ with a fixed set of parameters/coefficients, each polynomial-time computable, so that a potential complexity blowup traces back to the inherent complexity of $S$ and not to other external factors. 

\subsection{Contribution}
In this work, we show that any non-trivial linear ordinary differential equation (ODE) of the form  
\begin{equation*}
    \sum_{i=0}^{m} a_i y^{(i)} = \sum_{j=0}^{n} b_j u^{(j)}, \quad m,n \in \mathbb{N}, a_i,b_j\in\mathbb{C}, u \in \mathcal{C}^n([0,1], \mathbb{C}), y\in \mathcal{C}^m([0,1], \mathbb{C}),
\end{equation*}
exhibits complexity blowup on digital machines modeled by Turing machines. Here, $y^{(i)}$ and $u^{(j)}$ denote the $i$-th and $j$-th derivative of the univariate functions $y$ and $u$, respectively. Note that we implicitly assume for simplicity that $y(t)$ is defined for all $t\in [0,1]$, i.e., the ODE does not grow arbitrarily large at some $\tau$ satisfying $0 < \tau < 1$. However, we believe that this assumption is not necessary and can be removed by refining our proof strategy. Indeed, for linear ODEs with constant coefficients and smooth input signal $u$---a key scenario for our considerations---, the solution $y$ is a continuous function that is smooth on $(0,1]$ \cite{Doetschbook}. 
Finally, the ODE is non-trivial if $y$ cannot be written as $c_n u^{(n)}+\ldots+c_0 u$ for some $c_0, \dots, c_n\in\mathbb{C}$. 

In this setting, we establish a link between the algebraic structure of ODEs and their computational complexity measured in terms of accuracy. In particular, we show that there exists a smooth input signal $u$ and coefficients $a_i, b_j$, $i=1,\dots,m$, $j=1,\dots,n$, which can be computed efficiently, i.e., in polynomial time, such that the corresponding solution of the ODE is $\#P$-complete. Under the assumption $P\neq \mathit{NP}$, it follows that the solution is not efficiently, i.e., in polynomial time, computable. Furthermore, we show that this is an upper bound on the complexity of the solution: if the input signal can be computed in polynomial time, then the solution is at most in $\#P$. 
This extends the analysis of the first-order case ($m=n=1$) in \cite{Boche2021ComplexityBlowup} to higher-order equations ($m, n \in \mathbb{N}$). In both settings, the proof strategies build on foundational work on the computational complexity of real functions and operators \cite{kawamura_computational_nodate, KerIFriedman82}.  The key ingredient is Friedman's seminal result about complexity of integration that shows the existence of a smooth polynomial-time computable functions on real intervals whose integral/accumulation functions are $\#P$-complete \cite{FRIEDMAN198480}.

Beyond theoretical interest, our result forms a foundation on which further discussions about the complexity of emulating analog (electrical) circuits can be based. In particular, non-trivial ODEs basically comprise all relevant applications in electrical engineering and physics based on linear systems such as any linear time-invariant (LTI) system of arbitrary order, including, for instance, simple LR low-pass and RC low-pass electrical circuits \cite{Boche2021ComplexityBlowup}. Hence, the simulation of all these dynamics on digital hardware is prone to a complexity blowup.   

Moreover, the most basic description of neuronal dynamics is based on electrical circuits governed by linear ODEs systems such as the (leaky) integrate and fire (LIF) model \cite{gerstner_kistler_naud_paninski_2014, Gerstner_Kistler_2002}. To account for these potentially more complex behaviours, we also begin to extend our observations to systems of linear ODEs, i.e., equations of the form
\begin{equation*}
    \sum_{i=0}^{m} A_i\cdot y^{(i)} = \sum_{j=0}^{n} B_j\cdot u^{(j)}, \quad A_i,B_j\in\mathbb{C}^{d\times d}, u\in\mathcal{C}^n([0,1], \mathbb{C}^d), y\in \mathcal{C}^m([0,1], \mathbb{C}^d) 
\end{equation*}
with $m,n,d \in \mathbb{N}$. Due to the different underlying algebraic structures of the domain (compared to single ODEs), we cannot directly adapt the previous techniques to derive an analogous result. However, for systems with order $m=1$ and under further technical assumptions, which are, for instance, met by the LIF model, we can still prove the existence of complexity blowup in the above sense. Obtaining a full description of the complexity blowup phenomenon for general systems, while its occurrence is still expected, is left for future work. Similarly, it is an interesting question whether and to what degree complexity blowup arises in more general dynamics such as linear time-invariant differential-algebraic equations. Here, the increased intrinsic complexity of the dynamics suggest that complexity blowup may still naturally occur although certain properties of linear ODEs do not hold any more \cite{BrenanAlgDif, PetzoldDifAlg, ReissigBocheBarton02}.

\subsection{Related Work}
There is a rich body of literature on the computability and complexity of differential equations; a good introduction and summary are provided in \cite{brattka_computability_2021}. We briefly review the most relevant results for our work here.

With regard to ODEs, it has been shown that the solutions of computable ODEs, i.e. equations of the form $y' = f(t,y)$ with computable initial condition $y(0)$ and sufficiently regular (i.e. Lipschitz) righthand side $f$, are computable \cite{collins_effective_2008}. Thus, we will implicitly assume the computability of the solution throughout this work on linear ODEs. When considering complexity, it was shown that solutions can always be computed in PSPACE \cite[Theorem 3.3.2.]{brattka_computability_2021}, i.e., solved by a Turing machine using a polynomial amount (in the input size) of space or equivalently polynomial many tapes. 
Furthermore, there exist polynomial-time computable ODEs (i.e., instead of mere computability, one requires polynomial-time computability of the initial condition and righthand side) whose solution is PSPACE-complete \cite[Theorem 3.3.3.]{brattka_computability_2021}. This indicates that, in general, the solution of ODEs can be highly complex. In contrast, we show that by restricting to linear ODEs, the complexity of the solution is reduced to $\#P$-completeness.

Regarding PDEs, several works have analyzed the complexity of specific equations. These include the wave equation \cite{PourElNingZhong06}, the inhomogeneous Schrödinger equation \cite{WeihrauchZhong01}, and Navier-Stokes equations \cite{SunZhongZiegler}. It has also been shown that the solutions of certain linear PDEs can be PSPACE-computable \cite{SteinbergZiegler16}, again showing that by generalizing our setting (now from ordinary to partial differential equations but maintaining the linearity), the complexity of the solution can increase significantly.

\subsection{Outline}

We begin by introducing the necessary concepts in complexity theory and linear ODEs in Section \ref{sec:Prel}. Subsequently, in Section \ref{sec:main}, we state our results, starting with determining the exact complexity of first-order linear ODEs. Next, we extend this to higher orders by induction on the order, exploiting the algebraic structure of linear ODEs. Finally, we state a similar result for a subclass of first-order systems of linear ODEs. We conclude by discussing the practical relevance of these results, highlighting a potential application in neuronal dynamics, in Section \ref{sec:discussion} followed by the proofs of our results in Section \ref{sec:proofs}.

\subsection{Notation}
We write $\{0,1\}^\ast$ for the set of all finite binary strings over $\{0,1\}$ and denote by $|x|$ the length of a string $x\in\{0,1\}^\ast$. Similarly, $|A|$ denotes the cardinality of a set $A$.

We use $f:\mathcal{X}\rightharpoonup\mathcal{Y}$ to indicate that $f$ is a partial function from $\mathcal{X}$ to $\mathcal{Y}$. Moreover, for $\mathcal{X}\subset\mathbb{R}$ and $\mathcal{Y}\subset\mathbb{C}^d$, $\mathcal{C}^n(\mathcal{X},\mathcal{Y})$ denotes the class of $n$-times continuously differentiable functions from $\mathcal{X}$ to $\mathcal{Y}$. 

For a ring $\mathcal{R}$, $\mathcal{R}[X]$ denotes the ring of polynomials in the indeterminate variable $X$ with coefficients in $\mathcal{R}$. Additionally, for $P(X),Q(X)\in \mathcal{R}[X]$, we write $Q(X)\mid P(X)$ if $P(X)=Q(X)\cdot S(X)$ for some $S(X) \in\mathcal{R}[X]$, i.e., $P(X)$ divides $Q(X)$, and $Q(X)\nmid P(X)$ otherwise. Furthermore, $\mathcal{Z}(P)$ denotes the multiset of roots of $P(X)$.

\section{Preliminaries}\label{sec:Prel}

This section presents the necessary background from complexity theory and the theory of ODEs.

\subsection{Computable Analysis and Complexity Theory}

The Theory of Computability presented here is based on the computational model of the Turing machine \cite{Turing36Entscheidung}. This is an abstract machine that represents a natural and intuitive model of computation \cite{CooperBook, AroraBook}. In fact, by the Church-Turing Thesis \cite{Copeland2020CTT}, a problem can be solved by an intuitively effective method if and only if it can be solved by a Turing machine. Thus, by reasoning about Turing machines, we can obtain results on the intrinsic complexity of a given problem.

\subsubsection{Computable Analysis}

Classical computability theory is based on Turing machines, which operate in an inherently discrete domain \cite{Turing36Entscheidung, CooperBook, AroraBook}. Turing machines operate on strings on a tape over an alphabet such as $\{0,1\}$. Hence, they intrinsically compute partial (word) functions $ \{0,1\}^* \rightharpoonup \{0,1\}^*$, where $\{0,1\}^*$ denotes the set of all finite binary strings. Thus, Turing machines operate with inputs and outputs in $\{0,1\}^*$ . By introducing (partial) surjective functions $\{0,1\}^* \rightharpoonup \mathcal{X}$, one can consider Turing machines operating on any countable set $\mathcal{X}$, e.g., using binary numbers to encode natural numbers. Similarly, Turing machines computing binary functions $\{0,1\}^* \times \{0,1\}^* \rightharpoonup \{0,1\}^*$, i.e., Turing machines with two (or any finite number of) binary input strings, can be reduced to Turing machines operating on $\{0,1\}^* \rightharpoonup \{0,1\}^*$ by, for instance, introducing a separator symbol in the alphabet. Thus, we often assume Turing machines directly operating on $\mathbb{N}$ and/or $\mathbb{Q}$ for (notational) simplicity.

Crucially, a Turing machine cannot directly operate on infinite objects such as real numbers or real functions. However, we can extend classical computability theory to infinite objects by means of finite approximations \cite{Pour-El17Computability, Weihrauch00CompAnal, Ko91ComplTheoryRealFunc}. In particular, we consider Turing machines with a single tape and count the number of necessary computation steps to achieve a prescribed approximation accuracy given as input. Since only one computation step is possible per base time unit, this is the computation time. As already mentioned in the context of PSPACE, the Turing machine could also have multiple parallel tapes, which would reflect the necessary computation space.
\begin{definition}[Computable Real number] \label{def:computable_real_number}
    A real number $x \in \mathbb{R}$ is \emph{computable} if there exists a Turing machine $M:\mathbb{N} \to \mathbb{Q}$ such that 
    \begin{equation}\label{eq:computable_real_number}
               \left| x - M(n) \right| \leq 2^{-n} \quad  \forall n\in\mathbb{N}.
    \end{equation}
    In this case, we say that $M$ converges effectively to $x$, written as $M \xrightarrow{\text{eff}} x$.
\end{definition}
\begin{remark}
     We write $\mathbb{R}_c\subsetneq \mathbb{R}$ for the (dense) set of computable real numbers. A complex number $z = a + bi \in \mathbb{C}$ is computable if both $a,b\in\mathbb{R}_c$ are computable. Similarly, a matrix $A \in \mathbb{C}^{m\times n}$ is computable if all its entries are computable. We write $\mathbb{C}_c$ and $\mathbb{C}_c^{m\times n}$ for the set of computable complex numbers and computable complex matrices, respectively.
\end{remark}
To quantify the complexity of computing real numbers, the number of iterations to reach the approximation rate prescribed in \eqref{eq:computable_real_number} can be considered.
\begin{definition}[Polynomial-Time Computable Numbers]\label{def:polynomial_time_computable_number}
    A real number $x \in \mathbb{R}$ is \emph{polynomial-time computable} if it is computable by a Turing machine $M$ and if there is a polynomial $P(X)\in\mathbb{N}[X]$ such that \eqref{eq:computable_real_number} holds after at most $P(n)$ steps of $M$. 
    
    A complex number or matrix is \emph{polynomial-time computable} if its components are poly\-nomial-time computable. We write $\mathbb{R}_p$, $\mathbb{C}_p$, and $\mathbb{C}_p^{m\times n}$ for the set of polynomial-time computable real numbers, complex numbers, and complex matrices, respectively.
\end{definition}
To define computable functions, we employ Oracle Turing machines \cite{Ko91ComplTheoryRealFunc}; an equivalent characterization can be achieved via Weihrauch's type-two effectivity approach \cite{Weihrauch00CompAnal}. An Oracle Turing machine $M_{\boldsymbol{\cdot}}$ is simply an ordinary Turing machine equipped with the feature to query oracle tapes. Conceptually, one can think of an oracle $\gamma$ as a function from $\mathbb{N}$ to $\mathbb{Q}$, which is able to calculate the function value $\gamma(\cdot)$ and print it on its tape in a single operation. Hereby, an oracle may process a potentially impossible computations by digital means, e.g., an oracle $\gamma$ could represent the effective convergence to a non-computable real number $x \in \mathbb{R}$. Thus, informally speaking, the function value $f(x)$ of a computable real function $f$ might be computed (in a certain sense) by an associated Oracle Turing machine $M_{\boldsymbol{\cdot}}$ equipped with the oracle $\gamma$ denoted as $M_\gamma$. However, this is not feasible when considering only Turing machines representing digital computations. But the crucial point is that the actual machine $M_{\boldsymbol{\cdot}}$ is solely responsible for the computation of the function $f$, and the oracle $\gamma$ is only used to provide approximations to the input value $x$. Consequently, this allows us to separate the computation of the inputs from the computation of the function evaluation. Thus, Oracle Turing machines provide a tool to clearly assess the intrinsic computability and in a second step the complexity of functions.  

\begin{definition}[Computable Function]\label{def:computable_function}
    A function $f:\mathbb{R}\to\mathbb{R}$ is \emph{computable} if there exists an Oracle Turing machine $M_{\boldsymbol{\cdot}}:\mathbb{N}\to\mathbb{Q}$ such that for all $t\in\mathbb{R}$ and all oracles $\gamma:\mathbb{N}\to\mathbb{Q}$ with $\gamma\xrightarrow{\text{eff}}t$  we have 
    \begin{equation}\label{eq:computable_function}
         \left| f(t) - M_\gamma(n) \right| \leq 2^{-n} \quad \forall n\in\mathbb{N}.
    \end{equation}
    In this case, we write $M_{\boldsymbol{\cdot}} \xrightarrow{\text{eff}} f$.
\end{definition}

Analogously to polynomial-time computable numbers, we can define polynomial-time computable functions. However, there is a subtle but important issue not yet discussed. Simply assuming that an oracle effectively converges to a given input is not sufficient since the convergence can be slow which potentially artificially introduces a higher complexity. For instance, an oracle could produce a rational approximation $\tfrac{p}{q}$ of some $x \in [0,1]$ with accuracy $2^{-n}$, where $p,q$ have both binary/decimal length $2^{n}$. Then, the associated Turing machine would already require exponential time to just read the result of the oracle call. To resolve this issue one would need to measure (polynomial) complexity both relative to the accuracy of the output and the size of the oracle encoding the input. For our needs, it suffices to employ dyadic rational numbers instead of all rational numbers as a basis for computable real numbers, i.e., to substitute $M(n)$ by $M(n)/2^n$ in \eqref{eq:computable_real_number} and require that $M:\mathbb{N}\to \mathbb{Z}$ instead of $M:\mathbb{N}\to \mathbb{Q}$ \cite{KawamuraCookComplx2013}. This leads to the following definition. 
\begin{definition}[Polynomial-Time Computable Function] \label{def:polynomial_time_computable_function}
    A function $f:[0,1]\to\mathbb{R}$ is \emph{polynomial-time computable} if there exists an Oracle Turing machine $M_{\boldsymbol{\cdot}}:\mathbb{N}\to\mathbb{Q}$ such that for all $t\in[0,1]$ and all oracles $\gamma:\mathbb{N}\to\mathbb{Z}$ with $\tfrac{\gamma}{2^n}\xrightarrow{\text{eff}}t$  we have 
    \begin{equation}\label{eq:computable_function_2}
         \left| f(t) - M_{\gamma/2^n}(n) \right| \leq 2^{-n} \quad \forall n\in\mathbb{N}
    \end{equation}
    and there exists a polynomial $P(X)\in\mathbb{N}[X]$ such that \eqref{eq:computable_function_2} holds after at most $P(n)$ steps of $M_{\boldsymbol{\cdot}}$. 
\end{definition}

\begin{remark}
    A complex-valued function $f:\mathbb{R}\to\mathbb{C}$ is \emph{computable} if both its imaginary and its real parts are computable, i.e., there exists an Oracle Turing machine $M_{\boldsymbol{\cdot}}:\mathbb{N}\to\mathbb{Q}[i]:=\{p + iq \mid p,q\in\mathbb{Q}\}$ such that \eqref{eq:computable_function} holds. Analogously, $f$ is polynomial-time computable if both its imaginary and real parts are polynomial-time computable. 
    
    Moreover, a function $f:[0,1]\to\mathbb{C}^d$ with $d\in\mathbb{N}$ is called computable (resp. polynomial-time computable) if all its components $f_i:[0,1]\to\mathbb{C}$ for $1\leq i\leq d$ are computable (resp. polynomial-time computable). We write $\mathrm{Comp}([0,1], \mathbb{C}^d)$ and $\mathrm{Pol}([0,1], \mathbb{C}^d)$ for computable and polynomial-time computable functions $f:[0,1]\to\mathbb{C}^d$, respectively. If $d=1$, we write $\mathrm{Comp}([0,1])$ and $\mathrm{Pol}([0,1])$. This and the following definitions and theorems can be straightforwardly extended to functions on any compact interval $[a,b]$, however, for simplicity, we only consider functions and ODEs on the interval $[0,1]$. 
\end{remark}
\begin{remark}
     The class of polynomial-time computable functions is denoted by $\mathit{FP}$.
\end{remark}

Equipped with the notion of (polynomial-time) computable functions, we can now formally introduce the concept of a complexity blowup. 
\begin{definition}[Complexity Blowup]\label{def:complexity_blowup}
    Consider $\mathcal{F}:=\{f \mid f:[0,1]\to\mathbb{C}^d\}$ and an operator $T: \mathcal{F}\to \mathcal{F}$. We say that $T$ \emph{preserves computability} or \emph{polynomial-time computability} 
    if $T\left(\mathrm{Comp}([0,1];\mathbb{C}^d)\right)\subseteq \mathrm{Comp}([0,1];\mathbb{C}^d)$ or $T\left(\mathrm{Pol}([0,1];\mathbb{C}^d)\right)\subseteq \mathrm{Pol}([0,1];\mathbb{C}^d)$, respectively. A computability preserving operator $T$ exhibits \emph{complexity blowup} if $T\left(\mathrm{Pol}([0,1];\mathbb{C}^d)\right)\nsubseteq \mathrm{Pol}([0,1];\mathbb{C}^d)$. 
\end{definition}
\begin{remark}
    One needs to distinguish between an operator being (polynomial-time) computable and preserving (polynomial-time) computability. Although the latter follows from the former, the converse is false \cite{FRIEDMAN198480, KAWAMURA2015Complexity}.
\end{remark}

It turns out that the complexity of functions and the occurrence of a complexity blowup are linked to the classical complexity classes on binary strings, which we review next.

\subsubsection{Complexity Classes} \label{subsec:CompClass}

The most well-known complexity classes $P$ (solvable in polynomial time) and $\mathit{NP}$ (verifiable in polynomial time) are associated with decision problems, i.e., problems expressed by partial functions mapping from $\{0,1\}^\ast$ to $\{0,1\}$. Since our analysis treats function evaluation problems, considering decision problems, i.e., deciding membership, is no longer sufficient. Instead, a different approach is generally needed to assess the complexity of function evaluation, and a suitable tool is provided by counting problems expressed by partial function mappings $\{0,1\}^\ast \rightharpoonup \{0,1\}^\ast$. 

\begin{definition}[Function Problem]\label{def:function_problem}
    A \emph{function problem} is defined by a partial function $F:\{0,1\}^\ast \rightharpoonup \{0,1\}^\ast$. We say that a Turing machine $M$ \emph{computes} $F$, if for all $x\in\{0,1\}^\ast$ we have $M(x) = F(x)$. Moreover, $F$ is in
    \begin{itemize}
        \item the class $\mathit{FP}$, if it can be computed by a (deterministic) polynomial-time Turing machine.
        \item the class $\#P$, if there exist a (deterministic) polynomial-time Turing machine $M$ and a polynomial $P\in\mathbb{N}[X]$ such that for all $x\in\{0,1\}^\ast$ we have 
        \begin{equation*}
                    \nu(F(x) )= \left| \left\{y \in\{0,1\}^{P(|x|)} \mid M(x,y) = 1 \right\}\right|,
        \end{equation*}
        where $\nu: \{0,1\}^\ast \to \mathbb{N}$ is a bijection encoding strings into natural numbers. 
    \end{itemize}
\end{definition}
In terms of Oracle Turing machines on real numbers the classes $\mathit{FP}$  and $\#P$, which (to a certain degree) can be thought of as analogs of $P$ and $\mathit{NP}$ for function evaluation problems, can be characterized in the following way: $\mathit{FP}$ (as already described in Definition \ref{def:polynomial_time_computable_function}) denotes the class of functions that can be computed by a Function-oracle Turing machine in polynomial time, whereas $\#P$ denotes the class of functions that enumerate the number of computations that end in an accepting state of polynomial-time Function-oracle Turing machines. 
Hence, a problem in $\#P$ asks for the number of polynomial-time certificates for an input; in comparison, $\mathit{NP}$ asks whether such a certificate exists. 
Note that $\mathit{FP} \subseteq \#P$ holds, but it is unknown whether $\mathit{FP} = \#P$, i.e., whether all problems in $\#P$ can be solved in polynomial time. Since $\mathit{FP} = \#P$ implies $P = \mathit{NP}$ (or equivalently, $P \neq \mathit{NP}$ implies $\mathit{FP} \neq \#P$), it is commonly assumed that this is not the case. Moreover, the importance of the class $\#P$ is demonstrated by the polynomial hierarchy and Toda's theorem \cite{Toda91, AroraBook}. The polynomial hierarchy describes a hierarchy of complexity classes based on the following iteration: 
\begin{itemize}
    \item $\mathcal{S}_0:=P$ and $\mathcal{V}_0:=P$
    \item For all $k \in \mathbb{N}$ we have $\mathcal{S}_{k+1}:=P^{\mathcal{V}_k}$ and $\mathcal{V}_{k+1}:=\mathit{NP}^{\mathcal{V}_k}$, where $P^{\mathcal{V}_k}$ and $\mathit{NP}^{\mathcal{V}_k}$ denote the class of problems in $P$ and $\mathit{NP}$ when using an oracle for $\mathcal{V}_k$ problems, respectively.
\end{itemize}
Thus, for example, we have $\mathcal{V}_1=\mathit{NP}$ and $\mathcal{S}_2=P^{\mathit{NP}}$, i.e., the class of problems that can be solved in polynomial time using an oracle for $\mathit{NP}$ problems. Crucially, Toda's theorem states that the entire polynomial hierarchy is contained in $P^{\#P}$ \cite{Toda91}. Thus, if there were a polynomial-time algorithm for any $\#P$-complete problem, it would allow us to solve any problem in the polynomial hierarchy in polynomial time, which is commonly deemed highly unlikely.

\subsection{(Systems of) Linear Ordinary Differential Equations}

Next, we formally introduce linear ODEs, focusing on their algebraic structure.

\begin{definition}[Linear Ordinary Differential Equation]\label{def:linear_differential_equation}
    A \emph{linear ordinary differential equation} of order $(m,n)\in \mathbb{N}^2$ is of the form 
    \begin{equation*}
        \sum_{i=0}^{m} a_iy^{(i)} = \sum_{k=0}^{n} b_k u^{(k)},
    \end{equation*}
    where $a_i,b_j\in\mathbb{C}$, $a_m, b_n \neq 0$ are constants, $u\in \mathcal{C}^{n+1}([0,1], \mathbb{C})$ is the inhomogeneous generator, and $y\in \mathcal{C}^{m+1}([0,1], \mathbb{C})$ is the unknown function. Since $a_m\neq 0$, we may always assume without loss of generality that $a_m = 1$. Moreover, we write $\mathrm{SOL}_{m,n}((a_i)_{i=0}^{m},(b_k)_{k=0}^{n}): u\mapsto y$ for the solution operator of the ODE.
\end{definition}

We will express our main result about the complexity blowup of linear ODEs in terms of the properties of the characteristic polynomials, which we introduce next.

\begin{definition}[Characteristic Polynomial]\label{def:characteristic_polynomial_equation}
    Consider a linear ODE of order $(m,n)$, i.e. of the form 
    \begin{equation*}
        \sum_{i=0}^{m} a_iy^{(i)} = \sum_{k=0}^{n} b_k u^{(k)}.
    \end{equation*}
    We define the \emph{characteristic polynomials} $P_y(X)$ and $P_u(X)$ of the ODE as 
    \begin{equation*}
        P_y(X) = \sum_{i=0}^{m} a_i X^i \in \mathbb{C}[X] \quad \text{and} \quad P_u(X) = \sum_{k=0}^{n} b_k X^k \in \mathbb{C}[X].
    \end{equation*}
\end{definition}

Next, we briefly extend our framework to systems of linear ODEs that essentially couple several scalar ODEs.
\begin{definition}[System of linear ODE and its characteristic polynomial]\label{def:system_of_linear_odes}
    A \emph{$d$-system of linear ODEs} of order $(m,n)\in \mathbb{N}^2$ is of the form 
    \begin{equation*}
        \sum_{i=0}^{m} A_i\cdot y^{(i)} = \sum_{j=0}^{n} B_j\cdot u^{(j)}
    \end{equation*}
    where $A_i,B_j\in\mathbb{C}^{d\times d}$ are constant, $u\in\mathcal{C}^{n+1}([0,1], \mathbb{C}^d)$ is the inhomogeneous generator, and $y\in\mathcal{C}^{m+1}([0,1], \mathbb{C}^d)$ is the unknown function. We write $P_y(X)$, $P_u(X)$ for its \emph{characteristic polynomials} given by 
    \begin{equation*}
        P_y(X) = \sum_{i=0}^{m} A_i X^i \in \mathbb{C}^{d\times d}[X], \quad P_u(X) = \sum_{j=0}^{n} B_j X^j \in \mathbb{C}^{d\times d}[X]
    \end{equation*}
    and $\mathrm{SOL}_{m,n}^d((A_i)_{i=0}^{m},(B_j)_{j=0}^{n}): u\mapsto y$ for the solution operator.
\end{definition}

Finally, we describe the general setting we are interested in studying the complexity blowup phenomenon. In particular, we are interested in cases where the inputs and coefficients of an ODE are regular/of low complexity a priori, so that a potential increase in the complexity of the output is indeed generated by the intrinsic complexity of the problem.

\begin{definition}[Polynomial-Time Computable (System of) ODEs]\label{def:comp_system_ode}
    A system of ODEs of the form \begin{equation*}
        \sum_{i=0}^{m} A_i\cdot y^{(i)} = \sum_{j=0}^{n} B_j\cdot u^{(j)}
    \end{equation*}
    with $A_i,B_j\in\mathbb{C}^{d\times d}$ and $u\in\mathcal{C}^{n+1}([0,1], \mathbb{C}^d)$ is (polynomial-time) computable, if all coefficients $A_i,B_j$, all initial conditions $y^{(i)}(0)$ with $0\leq i\leq m-1$, and the inhomogeneous generator $u$ are (polynomial-time) computable.
\end{definition}

\section{Complexity Blowup in linear ODEs}\label{sec:main}

This section presents our results, providing an exact characterization of the complexity of linear ODEs. We begin by examining first-order ODEs, building on the work in \cite{Boche2021ComplexityBlowup}. Next, we extend this to higher-order ODEs by induction. We conclude by proving an analogous result for certain first-order systems of ODEs.

Subsequently, we will always assume that our ODEs are well-behaved in the sense of Definition~\ref{def:comp_system_ode}, i.e., the parameter and the input of the ODE are polynomial-time computable. The question we ask is whether, in this well-behaved setting, one can also expect the corresponding solution to be well-behaved in the sense of low complexity, i.e., polynomial-time computability. In other words, does the solution operator preserve polynomial-time computability? Since, by construction, the potential complexity of the solution cannot arise due to external factors (high-complexity inputs or parameters), but only by the implicit properties of the ODE, we can assess whether an actual increase in complexity from input to output signal occurs. 

The proofs of the subsequent results can be found in Section~\ref{sec:proofs}.

\subsection{Linear ODEs}

Our main result shows that the answer to the above question, namely, whether the solution to a polynomial-time computable ODE is also polynomial-time computable in general, is false, assuming that $\mathit{FP}\neq \#P$. We actually obtain the stronger result that for any linear ODE of order $(m,n)$ with $P_y(X)\nmid P_u(X)$, there is a smooth and polynomial-time computable inhomogeneous generator such that the solution is $\#P$-complete. Thus, an algorithm that could solve such a polynomial-time computable linear ODE in polynomial time could be exploited to solve every problem in the polynomial hierarchy (as described in Subsection \ref{subsec:CompClass}) in polynomial time (by virtue of the $\#P$-complete instance), including, for instance, $\mathit{NP}$-complete problems.

\begin{theorem}\label{thm:main}
    If $\mathit{FP}\neq \#P$, a polynomial-time computable linear ODE of order $(m,n)$ exhibits complexity blowup if and only if $P_y(X)\nmid P_u(X)$.
\end{theorem}

The condition $P_y(X)\nmid P_u(X)$ describes a property of the algebraic structure of the ODE. If it is violated, then the ODE is trivial in a certain sense. Informally---a rigorous treatment is provided in Subsection \ref{subsec:4.1}---, a polynomial $P(X)\in\mathbb{C}[X]$ of the form 
\begin{equation*}
    P(X)=\sum_{i=0}^na_iX^i
\end{equation*}
can be linked to the polynomial differential operator $P(\frac{d}{dt})$ as \begin{equation*}
    P\left(\frac{d}{dt}\right):=\sum_{i=0}^n a_i\frac{d^i}{dt^i}.
\end{equation*}
Then an ODE of order $(m,n)$ can be stated succinctly as  \begin{equation*}
    P_y\left(\frac{d}{dt}\right)y(t)=P_u\left(\frac{d}{dt}\right)u(t) \quad \text{for } t\in [0,1].
\end{equation*}
Thus, assuming $P_y(X)\mid P_u(X)$, we have $P_u(X)=P_y(X)\cdot Q(X)$ for some $Q(X)\in\mathbb{C}[X]$ and the ODE can now be written as 
\begin{equation*}
    P_y\left(\frac{d}{dt}\right)y=P_y\left(\frac{d}{dt}\right)Q\left(\frac{d}{dt}\right)u.
\end{equation*}
Hence, it is clear that $y:=Q\left(\frac{d}{dt}\right)u$ is a solution to the ODE, which is polynomial-time computable. Our result shows that for all feasible inhomogeneous generators, this is actually the only possibility for the solution to be polynomial-time computable.

Next, we present more details of the proof of Theorem \ref{thm:main}. It is essentially divided into two steps. First, the complexity of the solution of linear ODEs of order $(1,n)$ is explicitly analyzed.  
\begin{proposition}\label{prop:first_order_general}
    For any (1,n)-order polynomial-time computable ODE of the form 
    \begin{equation*}
            y'+a_0y=\sum_{j=0}^{n} b_ju^{(j)}
    \end{equation*}
    the solution $y$ is in $\#P$. Furthermore, there exists $u\in\mathcal{C}^\infty([0,1], \mathbb{C}) \cap \mathrm{Pol}([0,1])$ such that $\mathrm{SOL}_{1,n}(a_0,(b_j)_{j=0}^n)(u)$ is $\#P$-complete if and only if $P_u(-a_0)\neq0$.
\end{proposition}
Note that the condition $P_u(-a_0)\neq 0$ is equivalent to $P_y(X)\nmid P_u(X)$ since $P_y(X)\mid P_u(X)$ holds if and only if $\mathcal{Z}(P_y)\subseteq\mathcal{Z}(P_u)$ (c.f. Lemma~\ref{lem:polynomial_divisibility_roots}) and $-a_0$ is the only root of $P_y$.

The proof strategy for linear ODEs of order (1,1), treated in \cite{Boche2021ComplexityBlowup}, also extends to linear ODEs of order (1,n) with some necessary modifications. 
The main idea is to rely on the solution formula, which we will exploit to explicitly encode a high-complexity problem into the solution of a linear ODE.
\begin{proposition}\label{prop:first_order_solution}
    The solution of a linear ODE of order $(1,n)$, i.e. of the form 
    \begin{equation*}
        y' + a_0 y = \sum_{k=0}^{n} b_k u^{(k)},
    \end{equation*}
    is given by 
    \begin{equation*}
        y(t)=e^{-a_0 t}y(0)+\sum_{j=0}^{n} \int_0^t b_j e^{-a_0(t-\tau)} u^{(j)}(\tau) d\tau.
    \end{equation*}
\end{proposition}
The encoding of the high-complexity problem, in fact, a $\#P$-complete problem, is built on the following well-known result about the complexity of integration \cite{FRIEDMAN198480}. 
\begin{theorem}\label{thm:integral_sharp_p_complete}
    There exists a smooth polynomial-time computable function $f:[0,1]\to\mathbb{C}$ such that the function $g$ given by 
    \begin{equation*}
        g(t) = \int_0^t f(\tau) d\tau
    \end{equation*}
    is $\#P$-complete.
\end{theorem}
\begin{remark}
    Under the stronger assumption that $f$ is analytic, the statement no longer holds, i.e., in this case $g$ is indeed polynomial-time computable as well. Also note that this is not a statement about polynomial-time computability of an operator on continuous functions. It is a point-wise statement in the sense that even if $\mathit{FP} = \#P$ holds true there might be no general efficient way (in the sense of polynomial-time computability) to obtain a procedure to compute the integral from a procedure to compute the function; indeed, this was shown in \cite{KAWAMURA2015Complexity}. 
\end{remark}
Putting these two observations together, we obtain the characterization of the complexity of linear ODEs of order $(1,n)$ presented in Proposition \ref{prop:first_order_general}.

The second step in the proof of Theorem \ref{thm:main} is to extend the first-order result to arbitrary linear ODEs by (polynomial-time) reduction techniques exploiting connections between the roots of the characteristic polynomials and certain polynomial reductions. 
\begin{proposition}\label{prop:general_case}
    For any $(m,n)$-order polynomial-time computable ODE of the form 
    \begin{equation*}
        \sum_{i=0}^{m} a_i y^{(i)} = \sum_{j=0}^{n} b_j u^{(j)},
    \end{equation*}
    the solution $y$ is in $\#P$. Furthermore, there exists $u\in\mathcal{C}^\infty([0,1], \mathbb{C}) \cap \mathrm{Pol}([0,1])$ such that $\mathrm{SOL}_{m,n}((a_i)_{i=0}^m,(b_j)_{j=0}^n)(x)$ is $\#P$-complete if and only if $\mathcal{Z}(P_y)\nsubseteq \mathcal{Z}(P_u)$.
\end{proposition}
Finally, Theorem \ref{thm:main} is now a direct consequence of Proposition \ref{prop:general_case} together with the next observation linking polynomial divisibility with roots (see, e.g., \cite{pinter_book_2010}).
\begin{lemma}\label{lem:polynomial_divisibility_roots}
    Let $F$ be an algebraically closed field, and let $P(X), Q(X) \in F[X]$ be polynomials. Then we have
    \begin{equation*}
        P(X)\mid Q(X)\iff \mathcal{Z}(P) \subseteq \mathcal{Z}(Q).
    \end{equation*}
\end{lemma}

\subsection{Systems of linear ODEs}

We now extend our findings to systems of linear ODEs. This is hindered by the fact that $\mathbb{C}^{d\times d}[X]$, respectively $\mathbb{C}_p^{d\times d}[X]$, is neither algebraically closed nor commutative. 
Thus, we cannot directly adapt the techniques used for single ODEs. However, in certain simplified settings that more closely mirror the properties of single ODEs, we can still obtain similar results.

\begin{theorem}\label{thm:first_order_systems}
    If $\mathit{FP}\neq \#P$, a polynomial-time computable $d$-system of linear ODE of order $(1,n)$ 
    \begin{equation*}
        A_1 y' + A_0 y = \sum_{j=0}^{n} B_j u^{(j)}
    \end{equation*}
    with $A_1$ invertible, such that for all $j$ the commutativity condition $A_1^{-1}A_0\cdot A_1^{-1}B_j=A_1^{-1}B_j\cdot A_1^{-1}A_0$ holds, exhibits complexity blowup if and only if $ P_y(- A_1^{-1}A_0)\neq 0$.
\end{theorem}
\begin{remark}
    To prove Theorem~\ref{thm:first_order_systems}, we employ the same proof strategy as in Proposition~\ref{prop:first_order_general}, i.e., use the explicit solution formula and encode a high-complexity problem. In particular, the solution formula contains the terms 
    \begin{equation*}
        e^{-A_1^{-1}A_0(t-\tau)}B_ju^{(j)}(t).
    \end{equation*}
    In one dimension, the exponential and $B_j$ commute, which allows us to eliminate the exponential in all terms simultaneously. However, in the general case, this no longer holds, which means that we cannot simplify the terms, and the proof strategy no longer applies. However, the additional assumptions in Theorem~\ref{thm:first_order_systems} (artificially) create a setting where the exponential and the coefficients commute.
\end{remark}

Under the given assumptions, the proof of Theorem \ref{thm:first_order_systems} is now completely analogous to the single case treated in Proposition \ref{prop:first_order_general} since we can rely on an explicit solution formula by a straightforward extension of Proposition \ref{prop:first_order_solution}. 
\begin{proposition}\label{prop:first_order_solution_operator_system}
    The solution of a $d$-system of linear ODEs of order $(1,n)$
    \begin{equation*}
        A_1 y' + A_0 y = \sum_{j=0}^{n} B_j u^{(j)}
    \end{equation*}
    with $A_1$ invertible is given by 
    \begin{equation*}
        y(t)=e^{-A_1^{-1}A_0 t}y(0)+\sum_{j=0}^{n} \int_0^t e^{-A_1^{-1}A_0(t-\tau)} A_1^{-1} B_j u^{(j)}(\tau) d\tau,
    \end{equation*}
    where integration is understood component-wise. 
\end{proposition}
Crucially, we cannot generalize this result to higher-order systems of ODEs by the proof technique used in Proposition \ref{prop:general_case} since it again relies on the fact that the base field is algebraically closed.

Although we did not yet fully prove it, we believe that (a variant of) the general result in Theorem \ref{thm:main} also extends to higher-dimensional systems. In particular, it seems unlikely that the increased complexity induced by the coupling of multiple ODEs would result in a solution of lower complexity in general. Essentially, it suffices to obtain an analogous result to Proposition~\ref{prop:first_order_general}, i.e., an exact characterization of the occurrence of complexity blowup in first-order systems of ODEs, since higher-order systems can be reduced to first-order systems.

\section{Discussion}\label{sec:discussion}

Now, we discuss the details and general implications of our findings. Subsequently, we exemplarily apply our framework to the modeling of neuronal dynamics, which are an important tool in understanding/studying biological neural networks and to a certain degree inspired modern deep learning architectures/dynamics.  

\subsection{Interpretation and Contextualization of Results}

The results in Theorems \ref{thm:main} and \ref{thm:first_order_systems} show that complexity blowup generally occurs in a wide range of simple models, in particular, (systems of) linear ODEs. Thus, it also seems likely that similar findings extend to more complex models, so that complexity blowup may be considered a rather general phenomenon. Although this appears to be a severe limitation in simulating/solving certain problems, there are also reasons to believe that the practical impact is less significant. In fact, deriving mathematical guarantees for low-complexity preserving operators in fairly general settings---e.g., simple linear ODES with polynomial-time computable coefficients---is infeasible due to our worst-case analysis, resulting in complexity blowup. However, we wish to mention the following points.
\begin{enumerate}
    \item The results are not constructive, but only the existence of complexity blowup is shown, since the main auxiliary result in Theorem \ref{thm:integral_sharp_p_complete} (complexity blowup in integration) is also only an existence result. Therefore, strengthening the results qualitatively or quantitatively is a hard problem, since characterizing the admissible input signals that lead to complexity blowup is difficult.  
    \item In our framework, complexity blowup is measured with respect to the accuracy of a computation. More exactly, the computation time is measured relative to the desired output accuracy. Thus, even in the presence of a (theoretical) complexity blowup, it may be negligible if low accuracies are sufficient or the blowup kicks in only at high accuracies. 
    \item Finally, the theory describes the worst-case complexity, i.e., the potential for a complexity blowup based on the existence of certain polynomial-time computable input signals. Hence, obtaining mathematical guarantees for maintaining low complexity via the solution operator is in a sufficiently general setting not feasible since the complexity blowup inducing signal(s) may be admissible input(s). However, transitioning from a worst-case to a (more practicable) average-case analysis is difficult because it requires (beyond the characterization of complexity blowup inducing inputs) a suitable probability distribution over $\mathrm{Pol}([0,1];\mathbb{C}^d)$.       
\end{enumerate}
An interesting open problem arising from this discussion is whether there exist universal (to a certain degree) input signals that trigger complexity blowup simultaneously in distinct ODEs. In other words, are the complexity blowup inducing input signals specifically adapted to a specific ODE (as we construct them in our proofs) or does a common input signal exist for certain classes of ODEs?

\subsection{Case Study: Complexity Blowup in Neuronal Dynamics}

Modeling and analyzing neural dynamics (as displayed in, for instance, human brains) is a key element of neuroscience. On the one hand, the simulation of said dynamics enables their study in controlled environments. On the other hand, the power and efficiency of biological neural networks make them a viable computational framework. Indeed, unlocking the benefits of brain-like properties in computational models has not only led to current AI systems powered by deep learning and artificial neural networks (ANNs) but also potentially  offers further advantages. Going beyond the loose analogy of artificial neurons and biological neurons in ANNs, spiking neural networks (SNNs) are envisioned to be more closely aligned with the structure of computations in biological neural networks, thereby exploiting their intrinsic benefits, foremost their energy efficiency \cite{maass1997networks, Mehonic2024Neuromorphic, Rathi_neuromorphic_based_SNNs_2023, fono2025sustainableaimathematicalfoundations}.   

The adoption of spikes---asynchronous and point-like electrical pulses in the biological context---as means of communication between neurons represents the key innovation of SNNs \cite{gerstner_kistler_naud_paninski_2014, Gerstner_Kistler_2002}. The dynamic of spikes is governed (in the most basic model, still retaining biological plausibility) by a system of (linear) ODEs. Hence, a crucial question is what the computational cost of simulating the dynamics or turning the framework into a computational model on digital hardware actually is. Is it plausible to expect the advantages of the biological model in the digital computing paradigm? We approach an answer by examining the algorithmic complexity of approximating the neuronal dynamics on digital hardware via the introduced complexity analysis of linear ODEs.

\subsubsection{Leaky Integrate and Fire Dynamics} 
We consider a simple model class of neuronal dynamics still retaining biological plausibility, the \emph{Integrate-and-Fire} (IF) models \cite{Gerstner_Kistler_2002}. They generally consist of two components: (i) the time evolution of the \emph{potential} $V(t)$ of a neuron given a \emph{current} $I(t)$ and (ii) the spike generation mechanism through thresholding operations. In the \emph{Leaky IF} (LIF) model, a time-continuous dynamical system describes the evolution of the potential via a system of linear differential equations,
\begin{equation}\label{eq:LIF1}
    \begin{cases}
        \tau_m \frac{\mathrm{d} V}{\mathrm{d} t}(t) = - (V(t) - V_\text{rest}) + I(t),& \\[4pt]
        \tau_s \frac{\mathrm{d} I}{\mathrm{d} t}(t) = - I(t) + I_e(t),&        
    \end{cases}
\end{equation}
where $\tau_m, \tau_s >0$ are the \emph{membrane} and \emph{synaptic time constants}, respectively, and $I_e(t)$ denotes an \emph{external input current}. When $V$ reaches a certain value, the \emph{firing threshold} $\vartheta>0$, the neuron emits a spike, and subsequently $V$ is reset to the \emph{resting potential} $V_\text{rest}$ while the spike triggers postsynaptic currents in downstream neurons. The temporal dynamics in a layered network of spiking neurons are then typically (mathematically) abstracted by conceiving spikes as point processes localized in time expressed as a sum of Dirac delta distributions, the so-called \emph{spike train} $S(t)=\sum_s \delta(t-s)$, where $s$ iterates over the spike times of a neuron.

\subsubsection{Complexity Blowup in LIF Dynamics}

Our objective is to understand the complexity of emulating the dynamics in \eqref{eq:LIF1} on digital hardware by examining the potential occurrence of a complexity blowup in this context. Note that \eqref{eq:LIF1} constitutes a system of first-order linear ODEs with constant coefficients, driven by the input signal $I_e(t)$ and producing the output signal $(V(t),I(t))$. 
Recall that these dynamics describe the behaviour of LIF neurons only up to the point where the potential reaches the threshold. Therefore, it is necessary to simulate the dynamics/compute the threshold crossing to obtain the first spike time of a LIF neuron. The actual firing and reset mechanisms, which affect subsequent dynamics, are not even considered in this simplified formulation. Nevertheless, it turns out that complexity blowup can already occur in this reduced setting.

\begin{theorem}\label{thm:CompBlow}
    If $\mathit{FP}\neq \#P$, the LIF model in \eqref{eq:LIF1} on $[0,1]$ with polynomial-time computable parameters, i.e., coefficients and initial conditions, exhibits complexity blowup.
\end{theorem}
\begin{proof}
    Observe that we can rewrite the LIF model from \eqref{eq:LIF1} as 
    \begin{equation*}
        A_1 y' + A_0 y = B_0 u,
    \end{equation*}
    with
    \begin{equation*}
        u = \begin{pmatrix} 0 \\ I_e \end{pmatrix}, \quad 
        y = \begin{pmatrix} u \\ I \end{pmatrix}, \quad 
        A_1 = \begin{pmatrix} 1 & 0 \\ 0 & 1 \end{pmatrix}, \quad  
        A_0 = \begin{pmatrix} \tau_m^{-1} & -\tau_m^{-1} \\ 0 & \tau_s^{-1} \end{pmatrix}, \quad 
        B_0 = \begin{pmatrix} \tau_s^{-1} & 0 \\ 0 & \tau_s^{-1}, \end{pmatrix},
\end{equation*}
where we neglected constants terms (i.e., $V_{\text{rest}}$) without loss of generality. Now, the statement immediately follows from verifying the conditions in Theorem \ref{thm:first_order_systems}.
\end{proof}
This result demonstrates that simulating the dynamics of a LIF neuron on digital hardware can be computationally expensive (or even practically infeasible) when an accurate determination of spike times is required. However, since the above theorem relies on a non-constructive existence proof for the specific input signal, we can neither meaningfully characterize its properties nor identify the subset of inputs with analogous effects on the solution operator. Additionally, for a given configuration of a LIF neuron the complexity blowup in \eqref{eq:LIF1} may only occur at a time $\tau$ after the potential $V(\tau)$ already crossed the threshold. This implies that the complexity blowup has no actual impact on the computations in the LIF model since only the dynamics up to the threshold crossing are relevant for the spike times. 

Moreover, in practice, the admissible input signals to networks of LIF neurons are typically more constrained. Specifically, the input signals are simplified to (weighted) spike trains, i.e.,  weighted Dirac delta distributions. Although this represents an unrealistic theoretical abstraction of the variety of neuro-physiological inputs in biological neurons, it nevertheless preserves essential features of their dynamics. 
Within our complexity framework, we conclude that---in this restricted setting---two opposite but equally restrictive situation may arise. First, the input signals themselves may already possess high complexity, which renders the question of output complexity less meaningful. Conversely, by constraining to a subset of polynomial-time computable input signals (which are intended to approximate the Dirac delta in a practical scenario) or if the complexity blowup occurs after the threshold crossing, the operator may in fact be polynomial-time computable on this subset or restricted domain. Deepening our understanding and providing a clearer characterization of the complexity blowup phenomenon in this setting is an interesting open problem.

\section{Proofs}\label{sec:proofs}

In this section, we present the proofs of the results stated in Section~\ref{sec:main}.

\subsection{General Observations and Preliminaries}\label{subsec:4.1}

The set $\mathbb{C}_p$ is an algebraically closed field and $\mathbb{R}_p\subset\mathbb{C}_p$ is a subfield \cite{matsui_polynomial_2006}. Based on this structure, we collect some useful properties about the complexity of functions and operations employed in the subsequent analysis.
\begin{proposition}\label{prop:propertiesComp}
    \begin{enumerate}
        \item[]
        \item Let $\mathbb{A}\in\{\mathbb{R},\mathbb{C}, \mathbb{C}^{d\times d}\}$. Then we have 
        \begin{equation*}
            \left\{ x+y\mid x\in\mathbb{A}_p, y\in\mathbb{A}_p^c  \right\} \cup\left\{ x\cdot y\mid x\in\mathbb{A}_p^*, y\in\mathbb{A}_p^c \right\} \subseteq \mathbb{A}_p^c,
        \end{equation*}
        where $\mathbb{A}_p^*$ denotes the set of elements with a (multiplicative) inverse in $\mathbb{A}_p$.
        \item The set $\mathrm{Pol}([0,1], \mathbb{C}^d)$ is a $\mathbb{C}_p$-algebra.
        \item Let $f\in\mathrm{Pol}([0,1])\cap \mathcal{C}^{n+1}([0,1],\mathbb{C})$. Then $f^{(i)}\in\mathrm{Pol}([0,1])$ for all $i =1, \dots, n$.
        \item Let $z\in\mathbb{C}_p$. Then $f: [0,1] \to \mathbb{R}$ given by $f(x)=e^{z x}$ is in $\mathrm{Pol}([0,1])$. 
        \item  Let $f:[0,1]\to\mathbb{C}$ be $\#P$-complete, let $g\in\mathrm{Pol}([0,1], \mathbb C)$ and $0\neq z\in\mathbb{C}_p$ be polynomial-time computable. Then the function 
        \begin{equation*}
            h(t):=zf(t)+g(t)
        \end{equation*}
        is also $\#P$-complete.
        \item If $f:[0,1]\to\mathbb{C}$ is in $\#P$, then so is $\int_0^tf(t)dt$.
        \item If $f':[0,1]\to\mathbb{C}$ is $\#P$-complete, then $f$ is also $\#P$-complete.
    \end{enumerate}
    \begin{proof}
        \begin{enumerate}
            \item[]
            \item Assume $x+y=:z$ is in $\mathbb{A}_p$ with $x\in\mathbb{A}_p$ and $y\in\mathbb{A}_p^c$. Then $y=z-x$ is also in $\mathbb{A}_p$, which is a contradiction. Multiplication follows analogously.
            \item Addition was shown in \cite{Boche2021ComplexityBlowup}, multiplication follows analogously, and scalar multiplication is a special case of multiplication.
            \item This follows immediately from Lemma 4.1 in \cite{KerIFriedman82}
            \item By Euler's formula $e^{a+ib}=e^a(\cos(b)+i\sin(b))$, it suffices to show that the real exponential, sine, and cosine functions are polynomial-time computable. The real exponential case was shown in \cite{Boche2021ComplexityBlowup}, and the other cases are analogous.
            \item This quickly follows from the transitivity of Turing reductions.
            \item This was shown in \cite{kawamura_computational_nodate}.
            \item We can apply the algorithm used in the proof of Theorem 5.2 in \cite{KerIFriedman82}, to calculate $f'$ from $f$ in polynomial time. Thus, $f'$ polynomial-time reduces to $f$, and $f$ is $\#P$-hard. And by the previous proof, $f$ must also be in $\#P$, and hence $\#P$-complete.
        \end{enumerate}
    \end{proof}
\end{proposition}
To characterize linear ODEs, we have already introduced their characteristic polynomials in Definition \ref{def:characteristic_polynomial_equation}. To motivate the appearance of these characteristic polynomials more formally, consider linear differential operators of order $n$ with constant coefficients 
\begin{align*}
    &D:\{f: [0,1] \to \mathbb{C}\mid f\text{ is } n \text{-times differentiable}\} \to \{f: [0,1] \to \mathbb{C}\}, \\
    &f \mapsto \sum_{k=0}^{n} a_k f^{(k)}(t), \quad a_k\in \mathbb{C}, n\in \mathbb{N}. 
\end{align*}
We write $\mathcal{P}([0,1])$ and $\mathcal{P}_n([0,1])$ for the set of all differential operators with constant coefficients on $[0,1]$ and those of order $n$, respectively. Now we can associate with each $D\in \mathcal{P}_n(I)$ a specific polynomial $P_D(X) \in \mathbb{C}[X]$:
\begin{equation*}
    D=\sum_{k=0}^{n} a_k \frac{d^k}{dt^k} \implies P_D(X) = \sum_{k=0}^{n} a_k X^k.
\end{equation*}
By abuse of notation, we will often write $D=P_D(\frac{d}{dt})$ for simplicity so that a linear ODE as introduced in Definition \ref{def:linear_differential_equation} can be expressed via its characteristic polynomials as 
\begin{equation*}
    P_y\left(\frac{d}{dt}\right) y = P_u\left(\frac{d}{dt}\right) u.
\end{equation*}
Moreover, one easily verifies the following equalities, which we will repeatedly apply in the remainder: For $P(X),Q(X)\in \mathbb{C}[X]$ we have
\begin{align*}
    (P(X)+Q(X))\left(\frac{d}{dt}\right) &= P\left(\frac{d}{dt}\right) + Q\left(\frac{d}{dt}\right) \\ (P(X)\cdot Q(X))\left(\frac{d}{dt}\right) &= P\left(\frac{d}{dt}\right) \cdot Q\left(\frac{d}{dt}\right).
\end{align*}
Another tool we need for our analysis is the reduction of the degree of polynomials. For $y\in R$ and a polynomial $P(X)\in R[X]$ over a ring $R$ with
\begin{equation*}
    P(X)=\sum_{i=0}^m a_i X^i,\quad a_i\in R,
\end{equation*}
we set 
\begin{align}\label{eq:polred}
    \mathcal{R}_z[P](X)&:=\sum_{i=0}^{m-1} b_i(z) X^i \in R[X], \quad \text{ where } \nonumber\\
    b_i(z)&:=\begin{cases}
        a_m, &\text{ if } i=m-1 \\
        a_{i+1}+z\cdot b_{i+1}(z), & \text{otherwise}
    \end{cases}.
\end{align}
In particular, $\mathcal{R}_z[P](X)$ is the polynomial obtained through the polynomial division of $P(X)$ by $(X - z)$ with remainder term $b_{-1}(y):=a_0+z\cdot b_0(z)$, i.e.,
\begin{equation*}
    P(X) = (X - z) \cdot \mathcal{R}_z[P](X) + b_{-1}(z).
\end{equation*}
This can be easily verified by direct calculation. In this setting, we immediately obtain the following characterizations.
\begin{proposition}\label{prop:polynomial_reduction}
    Let $F$ be a field, and let $P(X)\in F[X]$ be a polynomial. 
    \begin{itemize}
        \item With the notation introduced above, for $\sigma \in F$ the following are equivalent: (i) $\sigma\in\mathcal{Z}(P)$, (ii) $ b_{-1}(\sigma) = 0$, (iii) $\sigma b_0(\sigma) = -a_0$.
        \item For $\sigma\in\mathcal{Z}(P)$ we have $\mathcal{Z}\left(\mathcal{R}_\sigma[P]\right)=\mathcal{Z}(P)\setminus\{\sigma\}$.
    \end{itemize}
\end{proposition}

\subsection{Proof of Complexity Blowup in Linear ODEs}

\begin{proof}[Proof of Proposition \ref{prop:first_order_general}]
    According to Theorem \ref{thm:integral_sharp_p_complete}, there exists a function $\tilde{f}\in\mathrm{Pol}([0,1])\cap \mathcal{C}^\infty([0,1],\mathbb{C})$ such that 
    \begin{equation*}
        \tilde{g}(t):=\int_0^t \tilde{f}(\tau)d\tau
    \end{equation*}
    is $\#P$-complete. Set $u(t):=e^{-a_0 t}\tilde{f}(t)$ and observe that $u\in \mathrm{Pol}([0,1])\cap \mathcal{C}^\infty([0,1],\mathbb{C})$ as well via Proposition \ref{prop:propertiesComp}.2 and \ref{prop:propertiesComp}.4. Since for $i \in \{1, \dots, n\}$ we have
    \begin{align*}
        u^{(i)}(t) &= \sum_{j=0}^{i} \binom{i}{j} (-a_0)^{i-j}e^{-a_0 t}\tilde{f}^{(j)}(t)=e^{-a_0 t}\sum_{j=0}^{n} \binom{i}{j} (-a_0)^{i-j}\tilde{f}^{(j)}(t),
    \end{align*}
    it follows that 
    \begin{align*}
        \sum_{i=0}^{n}\int_0^t b_ie^{-a_0(t-\tau)}u^{(i)}(\tau)\;d\tau&=\sum_{i=0}^{n}\int_0^t b_ie^{-a_0t}\sum_{j=0}^{n}\binom{i}{j}(-a_0)^{i-j}\tilde{f}^{(j)}(\tau)\;d\tau \\
        &=e^{-a_0t}\sum_{j=0}^{n}\left(\sum_{i=0}^{n}\binom{i}{j}(-a_0)^{i-j}b_i\right)\int_0^t \tilde{f}^{(j)}(\tau)\;d\tau \\
        &=e^{-a_0t}\sum_{j=0}^{n}\lambda(j)\int_0^t \tilde{f}^{(j)}(\tau)\;d\tau \\
        &=e^{-a_0t}\left(\sum_{j=1}^{n}\lambda(j)f^{(j-1)}(t)+\lambda(0)\tilde{g}(t)\right) \\
        &=e^{-a_0t}\sum_{j=0}^{n-1}\lambda(j+1)f^{(j)}(t)+e^{-a_0t}\lambda(0)\tilde{g}(t)
        \end{align*}
        with $\lambda:\mathbb{N}\to\mathbb{C}_p$ given by
        \begin{align*}
            \lambda(j):=\sum_{i=0}^{n}\binom{i}{j}(-a_0)^{i-j}b_i.
        \end{align*}
        Thus, by the solution formula in Proposition \ref{prop:first_order_solution} we obtain that 
        \begin{align*}
                \mathrm{SOL}_{1,n}(a_0,(b_j)_{j=0}^n)(u) &= e^{-a_0t}y(0) + \sum_{i=0}^{n}\int_0^t b_ie^{-a_0(t-\tau)}u^{(i)}(\tau)\;d\tau \\
                &= e^{-a_0t}y(0) + e^{-a_0t}\sum_{j=0}^{n-1}\lambda(j+1)f^{(j)}(t)+e^{-a_0t}\lambda(0)\tilde{g}(t).
            \end{align*}
        Note that Proposition \ref{prop:propertiesComp}.2, \ref{prop:propertiesComp}.3, and \ref{prop:propertiesComp}.4 imply that 
        \begin{equation*}
                e^{-a_0t}\lambda(0), \;e^{-a_0t}\sum_{j=0}^{n-1}\lambda(j+1)f^{(j)}(t),\; \text{ and }\;e^{-a_0t}y(0)\in\mathrm{Pol}([0,1]).
        \end{equation*}
        Hence, $\mathrm{SOL}_{1,n}(a_0,(b_j)_{j=0}^n)(u)$ is $\#P$-complete if 
        \begin{equation*}
            0 \neq \lambda(0) = \sum_{i=0}^{n} \binom{i}{0}(-a_0)^{i-0}b_i = \sum_{i=0}^{n} b_i(-a_0)^i = P_u(-a_0)
        \end{equation*}
        follows from Proposition \ref{prop:propertiesComp}.5.
        Thus, it remains to show that if $P_u(-a_0)=0$, then the solution is always polynomial-time computable. Observe that by definition, $P_y(X)$ is a polynomial of degree one with a single root $-a_0$. Therefore, assuming $P_u(-a_0)=0$ implies that $\{-a_0\}=\mathcal{Z}(P_y) \subset \mathcal{Z}(P_u)$, which by Lemma \ref{lem:polynomial_divisibility_roots} is equivalent to $P_y(X)\mid P_u(X)$. Thus, there exists $Q(X)\in\mathbb{C}[X]$ with $P_u(X) = P_y(X)\cdot Q(X)$ so that the ODE can be rewritten as
        \begin{equation}\label{eq:eq1}
            P_y\left(\frac{d}{dt}\right)y = P_u\left(\frac{d}{dt}\right) u= (P_y(X)\cdot Q(X))\left(\frac{d}{dt}\right)u = P_y\left(\frac{d}{dt}\right)\left(Q\left(\frac{d}{dt}\right)u\right).
        \end{equation}
        Comparing the lefthand and righthand side shows that $y_p:=Q\left(\frac{d}{dt}\right)u$ is a particular solution of the ODE, which is polynomial-time computable by the structural properties of $\mathrm{Pol}([0,1])$ given in Proposition \ref{prop:propertiesComp}.2 and \ref{prop:propertiesComp}.3 for $u\in \mathrm{Pol}([0,1])\cap \mathcal{C}^{n+1}([0,1],\mathbb{C})$. Since the general solution of the ODE can be decomposed into the homogeneous solution $y_h$ and a particular solution $y_p$ we get
        \begin{equation}\label{eq:eq2}
            \mathrm{SOL}_{1,n}(a_0,(b_j)_{j=0}^n)(u) = y_h + y_p.
        \end{equation}
        Note that by Proposition \ref{prop:first_order_solution} the homogeneous solution is given by $y_h(t)=e^{-a_0 t}y(0)$, i.e., $y_h$ is polynomial-time computable via Proposition \ref{prop:propertiesComp}.2 and \ref{prop:propertiesComp}.4, and thereby $y_h + y_p\in \mathrm{Pol}([0,1])$.

        Finally, note that by applying Proposition~\ref{prop:propertiesComp}.6 to the solution formula from Proposition~\ref{prop:first_order_solution} it follows immediately that $y$ is at most in $\#P$, independent of $P_y$.
\end{proof}

\begin{proof}[Proof of Proposition \ref{prop:first_order_solution}]
    We consider the first-order ODE \begin{align*}
        y' + a_0 y = \sum_{j=0}^{n} b_j u^{(j)}
    \end{align*}
    where $a_0, b_j\in\mathbb{C}_p$ and $u\in\mathrm{Pol}([0,1])$ is in $\mathcal{C}^{n+1}([0,1], \mathbb{C})$. By multiplying with $e^{a_0 t}$, we get \begin{align*}
        e^{a_0 t} y'(t) + a_0 e^{a_0 t} y(t) = \frac{d}{dt}(e^{a_0 t} y(t)) = \sum_{j=0}^{n} b_j e^{a_0 t} u^{(j)}(t).
    \end{align*}
    Taking the integral, we get: \begin{align*}
        e^{a_0 t} y(t) - y(0) = \sum_{j=0}^{n} \int_{0}^{t} b_j e^{a_0\tau} u^{(j)}(\tau) d\tau.
    \end{align*}
    By rearranging, we obtain the desired formula.
\end{proof}

\begin{proof}[Proof of Proposition \ref{prop:general_case}]
    First, observe that for $u\in \mathrm{Pol}([0,1])\cap \mathcal{C}^n([0,1], \mathbb{C})$, if $P_y(X)\mid P_u(X)$, then the corresponding output of the ODE is polynomial-time computable by the same argument as in the proof of Proposition \ref{prop:first_order_general}, in particular \eqref{eq:eq1} and \eqref{eq:eq2}. Hence, by Lemma \ref{lem:polynomial_divisibility_roots} we know that $\mathcal{Z}(P_y) \subset \mathcal{Z}(P_u)$ implies that for $u\in \mathrm{Pol}([0,1])\cap \mathcal{C}^n([0,1],\mathbb{C})$ the corresponding solution $y = \mathrm{SOL}_{m,n}((a_i)_{i=0}^m,(b_j)_{j=0}^n)(u) \in \mathrm{Pol}([0,1])$ is polynomial-time computable. 

    Thus, assume that $\mathcal{Z}(P_y) \nsubseteq \mathcal{Z}(P_u)$, i.e., there exists $\tilde{\sigma}\in\mathcal{Z}(P_y) \setminus \mathcal{Z}(P_u)$. We show that there exists $u\in\mathcal{C}^\infty([0,1],\mathbb{C}) \cap \mathrm{Pol}([0,1])$ such that $\mathrm{SOL}_{m,n}((a_i)_{i=0}^m,(b_j)_{j=0}^n)(u)$ is $\#P$-complete by induction on $m$. The base-case $(1,n)$ was proven in Proposition \ref{prop:first_order_general}. Therefore, assume that the claim holds for all polynomial-time computable ODEs of order $(k,n)$ with $k<m$ and consider a polynomial-time computable ODE of order $(m,n)$ 
    \begin{equation*}
        \sum_{i=0}^{m}a_iy^{(i)}=\sum_{j=0}^{n}c_j u^{(j)},\quad a_i,c_i\in\mathbb{C}_{p}.
    \end{equation*}
    We now make the following observation: If there exist $b_1, \dots,b_{m-1},\sigma\in\mathbb{C}_{p}$ such that 
    \begin{equation}\label{eq:decomposition_condition}
        \sum_{i=0}^{m}a_iy^{(i)}=\sum_{i=0}^{m-1}b_i(y'-\sigma y)^{(i)},
    \end{equation}
    then $\tilde{y}:=y'-\sigma y$ is a solution of the polynomial-time computable ODE 
    \begin{equation}\label{eq:eq3}
        \sum_{i=0}^{m-1}b_i\tilde{y}^{(i)}=\sum_{j=0}^{n}c_j u^{(j)}.
    \end{equation}
    Hence, by the induction hypothesis we find $u\in\mathrm{Pol}([0,1]) \cap \mathcal{C}^\infty([0,1],\mathbb{C})$ such that $\tilde{y} = \mathrm{SOL}_{m-1,n}((b_i)_{i=0}^{m-1},(c_j)_{j=0}^n)(u)$ is $\#P$-complete provided that $\mathcal{Z}(\tilde{P}_y) \nsubseteq \mathcal{Z}(P_u)$ with $\tilde{P}_y(X):=\sum_{i=0}^{m-1}b_iX^i$ being the characteristic polynomial of the lefthand side corresponding to the modified ODE \eqref{eq:eq3}.
    Then, by Proposition~\ref{prop:propertiesComp}.7, $y$ is also $\#P$-complete. Note that, independently of the algebraic structure, we can also apply this construction to immediately show that the solution of the ODE is at most $\#P$.

    It is left to prove the existence of $b_1, \dots,b_{m-1},\sigma$ with the given properties and the fact that $\mathcal{Z}(\tilde{P}_y) \nsubseteq \mathcal{Z}(P_u)$ (or equivalently $\tilde{P}_y(X) \nmid P_u(X)$) holds. Unfolding the condition in \eqref{eq:decomposition_condition} gives
    \begin{equation*}
        a_my^{(m)}+\sum_{i=1}^{m-1}a_iy^{(i)}+a_0y = b_{m-1}y^{(m)} + \sum_{i=1}^{m-1}(b_{i-1}-\sigma b_i)y^{(i)}-\sigma b_0 y.
    \end{equation*}
    Hence, comparing the coefficients yields  
    \begin{equation*}
        a_m=b_{m-1},\quad b_{i-1}=a_i+\sigma b_i \text{ for } i=1,\dots,m-1, \quad \text{and} \quad \sigma b_0=-a_0.
    \end{equation*}
    Recalling the definition of polynomial reduction in \eqref{eq:polred} and its properties in Proposition \ref{prop:polynomial_reduction}, we notice that the condition is satisfied for $\tilde{P}_y(X)=\mathcal{R}_\sigma[P_y](X)$ where $\sigma\in\mathcal{Z}(P_y)$ and additionally $b_i, \sigma \in \mathbb{C}_p$ since $\in \mathbb{C}_p$ is an algebraically closed field (note that $\sigma$ is in $\mathbb{C}_p$, since it is a root of a polynomial in $\mathbb{C}_p[X]$). Since $|\mathcal{Z}(P_y)|=\text{deg}(P_y)=m\geq 2$ we can in particular choose $\sigma \neq \tilde{\sigma}$ so that again by Proposition \ref{prop:polynomial_reduction} 
    \begin{equation*}
        \mathcal{Z}(\tilde{P}_y)= \mathcal{Z}(\mathcal{R}_\sigma[P_y]) = \mathcal{Z}(P_y) \setminus \{\sigma\},
    \end{equation*}
    i.e., $\tilde{\sigma} \in \mathcal{Z}(\tilde{P}_y) \setminus \mathcal{Z}(P_u)$ or in other words $ \mathcal{Z}(\tilde{P}_y) \nsubseteq \mathcal{Z}(P_u)$. 
\end{proof}

\begin{proof}[Proof of Theorem \ref{thm:main}]
    Proposition \ref{prop:general_case} (together with Lemma \ref{lem:polynomial_divisibility_roots}) characterizes the complexity of the solution of a polynomial-time computable linear ODE. In particular, there exists $u\in\mathrm{Pol}([0,1]) \cap \mathcal{C}^\infty([0,1],\mathbb{C})$ such that the corresponding solution is $\#\text{P}$-complete if and only if the characteristic polynomials of the ODE satisfy $P_y(X)\nmid P_u(X)$. Thus, assuming $\mathit{FP}\neq \#P$, the solution is not polynomial-time computable, i.e., the ODE exhibits complexity blowup.    
\end{proof}

\subsection{Proof of Complexity Blowup in Systems of linear ODEs}

We can immediately extend our proof technique from Proposition \ref{prop:first_order_general} to systems of linear ODEs, assuming that the coefficients (now in the form of matrices) of the ODE in question commute. 
\begin{proof}[Proof of Theorem \ref{thm:first_order_systems}]
    By the commutativity conditions, we can rewrite the solution formula from Proposition~\ref{prop:first_order_solution_operator_system} as 
    \begin{equation*}
        y(t)=e^{-A_1^{-1}A_0t}y(0)+\sum_{j=0}^n\int_0^tA_1^{-1}B_je^{-A_1^{-1}A_0(t-\tau)}u^{(j)}(\tau)d\tau.
    \end{equation*}
    We can then apply an analogous construction as in the proof of Proposition~\ref{prop:first_order_general}, using the matrix exponential instead of the scalar exponential. Note that the matrix exponential (considered as a function) with a polynomial-time computable coefficient matrix---which is the case under the given assumptions---is a polynomial-time computable function. This follows from a direct extension of the scaler case in Proposition \ref{prop:propertiesComp}.4.
\end{proof}

\section*{Acknowledgments}
A. Fono and N. Wedlich were supported by the project ``Next Generation AI Computing (gAIn)'', funded by the Bavarian Ministry of Science and the Arts and the Saxon Ministry for Science, Culture, and Tourism.

This work of H. Boche was supported in part by the Federal Ministry for Research, Technology and Space of Germany (BMFTR) in the programme of ”Souverän.Digital.Vernetzt”, joint project 6G-life, project identification number 16KISK002. H. Boche was also partially supported by the project ``Next Generation AI Computing (gAIn)'', funded by the Bavarian Ministry of Science and the Arts and the Saxon Ministry for Science, Culture, and Tourism.

G. Kutyniok was supported in part by the Munich Center for Machine Learning (BMFTR) as well as the German Research Foundation under Grants DFG-SPP-2298, KU 1446/31-1 and KU 1446/32-1. She also acknowledges support by the Konrad Zuse School of Excellence in Reliable AI (DAAD) and the project "Next Generation AI Computing (gAIn)", which is funded by the Bavarian Ministry of Science and the Arts and the Saxon Ministry for Science, Culture and Tourism.

\printbibliography
\end{document}